\documentclass{article}

    \usepackage[preprint, nonatbib]{neurips_2020}

\usepackage[utf8]{inputenc} 
\usepackage[T1]{fontenc}    
\usepackage{hyperref}       
\usepackage{url}            
\usepackage{booktabs}       
\usepackage{amsfonts}       
\usepackage{nicefrac}       
\usepackage{microtype}      
\usepackage{todonotes}
\usepackage{caption,subcaption}
\usepackage{amsmath,amssymb, amsthm}
\usepackage{appendix}
\usepackage[ruled,vlined, linesnumbered]{algorithm2e}
\SetKwInput{KwInput}{Input}                
\SetKwInput{KwOutput}{Output} 
\newtheorem{theorem}{Theorem}
\newtheorem{definition}{Definition}

\newtheorem{lemma}{Lemma}
\newtheorem{prob}{Problem}
\newtheorem{assumption}{Assumption}

\usepackage{arydshln}
\usepackage[load-configurations = abbreviations]{siunitx}
\makeatletter
\newcommand{\printfnsymbol}[1]{%
  \textsuperscript{\@fnsymbol{#1}}%
}
\usepackage{tikz,pgfplots}
\usepackage{xcolor}
\usepackage{accents}
\pgfplotsset{compat=newest}
\usepgfplotslibrary{groupplots}
\usepgfplotslibrary{dateplot}
\usetikzlibrary{arrows,automata}
\pgfplotsset{compat=1.8}
\usepackage{pgfplotstable}
\usepgfplotslibrary{groupplots}

\title{Blending Controllers via Multi-Objective Bandits}

\author{
  Parham Gohari\thanks{Equal contribution} \\
  Department of Electrical and Computer Engineering\\
  University of Texas at Austin \\
  \texttt{pgohari@utexas.edu} \\
   \And
   Franck Djeumou\printfnsymbol{1} \\
   Department of Electrical and Computer Engineering \\
   University of Texas at Austin \\
   \texttt{fdjeumou@utexas.edu} \\
   \AND
   Abraham P. Vinod \\
   Oden Institute for Computational Engineering and Sciences \\
   University of Texas at Austin \\
   \texttt{aby.vinod@gmail.com} \\
   \And
   Ufuk Topcu \\
   Department of Aerospace and Mechanical Engineering \\
   University of Texas at Austin \\
   \texttt{utopcu@utexas.edu}
}

\begin{document}

\maketitle

\begin{abstract}

Safety and performance are often two competing objectives in sequential decision-making problems.
Existing performant controllers, such as controllers derived from reinforcement learning algorithms, often fall short of safety guarantees.
On the contrary, controllers that guarantee safety, such as those derived from classical control theory, require restrictive assumptions and are often conservative in performance.
Our goal is to blend a performant and a safe controller to generate a single controller that is safer than the performant and accumulates higher rewards than the safe controller.
To this end, we propose a blending algorithm using the framework of contextual multi-armed multi-objective bandits.
At each stage, the algorithm observes the environment's current context alongside an immediate reward and cost, which is the underlying safety measure.
The algorithm then decides which controller to employ based on its observations.
We demonstrate that the algorithm achieves sublinear Pareto regret, a performance measure that models coherence with an expert that always avoids picking the controller with both inferior safety and performance.
We derive an upper bound on the loss in individual objectives, which imposes no additional computational complexity.
We empirically demonstrate the algorithm's success in blending a safe and a performant controller in a safety-focused testbed, the Safety Gym environment.
A statistical analysis of the blended controller's total reward and cost reflects two key takeaways:
The blended controller shows a strict improvement in performance compared to the safe controller, and it is safer than the performant controller.

\end{abstract}

\section{Introduction}
\label{sec:introduction}

Designing autonomous systems that are both safe and well-performing is a significant challenge in artificial intelligence research.
On the one hand, reinforcement learning algorithms offer sophisticated controllers that perform a given task efficiently, yet they often fall short of safety guarantees \cite{zhu2019inductive, alshiekh2018safe}.
On the other hand, controllers that offer safety guarantees, \textit{e.g.}, controllers derived from classical control theory, are often conservative in performance \cite{sadraddini2018formal, fulton2018safe, chow2018lyapunov, berkenkamp2017safe}.
An intuitive solution might be to blend a safe and a performant controller to obtain a single controller that is both safe and performant.
In this paper, we investigate the possibility of blending such controllers.\par

We define ``blending controllers'' as learning a switching strategy between a given safe and a given performant controller using the observations and feedback from the environment.
We assume that the environment dynamics are unknown, and that, upon taking an action, the environment issues the agent with a feedback vector containing a one-step reward and an auxiliary cost that measures the safety of the action.
For example, the environment might issue a cost whenever the agent at the proximity of an obstacle.
The framework of auxiliary costs for measuring safety is conventional in reinforcement learning algorithms that respect safety~\cite{altman1999constrained, chow2017risk, dalal2018safe}, and is not limited to this work.
The performant controller achieves a higher expected total reward than the safe controller, whereas the safe controller has a lower expected total cost.\par 

Switching between the safe and the performant controller without proper measures may take the agent to a state that neither the safe controller renders as safe nor had it been experienced by the performant controller's underlying reinforcement learning algorithm.
We require a blending algorithm to justify its choice of controllers according to the \textit{Pareto dominance relationship}, \textit{i.e.}, we require the algorithm to avoid choosing a controller with both inferior safety and performance measures.  

We formally state the problem of blending controllers as follows: Fix a safe and a performant controller.
Let the environment associate every state-action pair with a two-dimensional feedback vector consisting of the one-step reward and safety measure.
Consider an expert who always avoids choosing the controller with a Pareto dominated feedback vector.
Then, design an online learning algorithm whose cumulative deviations from the expert's choice converges to zero on average.

In this paper, we propose a solution to the problem of blending controllers using \textit{contextual multi-objective multi-armed bandit algorithms} \cite{lattimore2018bandit}, wherein the safe and performant controllers are the arm set, performance and safety are the objectives, and the observations from the environment are the context.
The multi-objective formulation enables the algorithm to consider each safety requirement as a single objective, which simplifies the modeling of safety requirements.
For example, in autonomous driving, we may consider keeping the vehicle centered in the lane as one objective and avoiding obstacles as another.
Moreover, the algorithm is compatible with any number of input controllers. 
Therefore, for the example of autonomous driving, the algorithm may employ multiple safe controllers, each of which is safe with respect to some safety requirement.
We utilize the above formulation to develop an algorithm that solves the problem of blending controllers.

The main contributions of this work are as follows: 
\begin{itemize}
    \item \emph{Propose a novel contextual multi-armed multi-objective bandit algorithm for blending controllers.}
    The algorithm maintains an optimistic estimate of the next-step feedback for every arm. 
    These estimations, on which the algorithm bases its choice of arms, become more accurate as time progresses.
    The algorithm then picks the arm with the smallest estimated loss in individual objectives, and we show that such a decision rule leads to picking an arm whose estimated next-step feedback vector is not Pareto dominated by any other arm.
    
    \item \emph{Demonstrate that the algorithm's cumulative deviations from the expert's choice converges to zero on average.}
    We use the notion of \textit{Pareto regret}~\cite{drugan2013designing} to penalize the agent whenever it chooses a Pareto dominated arm.
    We then show that the Pareto regret of the algorithm is sublinear, which implies that the average Pareto regret converges to zero asymptotically.
    
    \item \emph{Establish a probabilistic bound on the average maximal loss in individual objectives.}
    The average maximal loss in individual objectives is a more intuitive performance measure for the problem of blending controllers than Pareto regret, a conventional performance score for multi-objective bandit algorithms. 
    The bound is directly computed from the estimates that the algorithm maintains; therefore, it imposes no additional computational complexity and can be computed on-the-fly.

\end{itemize}


We use \textit{Safety Gym}~\cite{raybenchmarking},
a testbed for reinforcement learning algorithms that respect safety, to
demonstrate the algorithm's effectiveness in blending controllers. 
In our experiments, we cover three levels of task and safety complexities in Safety Gym environments.
We construct the performant and the safe controllers using deep reinforcement learning methods \cite{mnih2015human}.
We generate the context for the proposed bandit algorithm using the action values estimated by the underlying neural networks.
In each environment, an analysis of the statistics of the reward and cost of the blended controller confirms that the blended controller shows a significant improvement in its safety when compared to the performant controller and in its performance when compared to the safe controller.

The rest of this paper proceeds as follows: In Section \ref{sec:preliminaries}, we fix the notation and definitions used throughout the paper followed by the problem statement. In Section \ref{sec:theory}, we introduce the algorithm that we propose for blending controllers as well as the theoretical developments. We discuss the numerical results in Section \ref{sec:numerical-examples}, and review the related works in Section \ref{sec:related_work}. We provide conclusions and directions for future research in Section \ref{sec:conclusions}. Finally, we take a step back from the technicalities and discuss the potential societal impacts of this work in Section \ref{sec:broader impacts}.

\section{Preliminaries and problem statement}
\label{sec:preliminaries}

\subsection{Notation} We denote the set of real numbers by $\mathbb{R}$, non-negative reals by $\mathbb{R}_+$, and natural numbers by $\mathbb{N}$.
For any $n\in\mathbb{N}$, $[n] := \{1,\dots,n\}$.
Let $v\in\mathbb{R}^n$ and $i\in[n]$, then $v^\top$ is the transpose and $(v)_i$ is the $i^\text{th}$ component of $v$.
For any $u,v \in \mathbb{R}^n$, the inner product of $u$ and $v$ is denoted by $u \cdot v$. Let $v\in\mathbb{R}^n$ and $W\in\mathbb{R}^{n\times n}$, then, $\|v\|_{W}$ is the matrix norm of $v$ with respect to $W$, \textit{i.e.,} $\|v\|_{W}^2 := v^\top W v$, and $\|v\|_2$ is the second norm of $v$.

Let $u,v\in\mathbb{R}^n$, then
$u$ is said to \textit{Pareto dominate} $v$, denoted $u
\succ v$ if and only if, for all $i\in[n]$, we have that
$u_i \ge v_i$ and there exists $j\in[n]$ such that $u_j$ is strictly greater than $v_j$.
We use notation $u \not \succ v$ if $u$ is \textit{not Pareto dominated} by $v$, \textit{i.e.,} $v \succ u$ or there exists $i,j\in[n]$ such that $i\neq j$, $v_i > u_i$, and $v_j < u_j$.


\subsection{Contextual multi-objective bandits}

In this section, we establish the definitions corresponding to c\textit{ontextual multi-armed multi-objective bandits}.
We denote the arm set of the bandit by $\mathcal{X}$, the
environment state space by $\mathcal{Z}$, and the learning
horizon by $T$. At any stage $t\in[T]$, a context vector $\Psi_t\in \mathbb{R}^d$
characterizes the agent's observation of the environment.
The context vector is defined using a known feature mapping
$\psi: \mathcal{Z}\times\mathcal{X}\mapsto\mathbb{R}^d$.
Specifically, let $z_t\in\mathcal{Z}$ be the current state of the environment and $x_t\in\mathcal{X}$ be the arm picked by the algorithm, then $\Psi_t = \psi(z_t,x_t)$.
Upon pulling an arm, the environment issues the
agent with a feedback vector $y_t\in\mathbb{R}^m$. The
feedback consists of $m$ objectives that the bandit seeks to
simultaneously optimize. Without loss of generality, we assume that the objective measurements are normalized such that each unit measurement has equal importance amongst all objectives.


We use the notion of \emph{Pareto
regret}~\cite{drugan2013designing} as the underlying performance measure in multi-objective bandits to penalize the agent whenever it picks an arm whose corresponding feedback vector is Pareto dominated by another arm.
We also introduce an additional performance measure, the \emph{cumulative maximal
loss}, which provides a finer criterion to distinguish between the set of non-dominated arms.
In the following definition, we define \textit{Pareto suboptimality gap} \cite{2019arXiv190512879L} and the maximal loss in individual objectives, which we later use to compute the Pareto regret and cumulative maximal
loss of the algorithm.
\begin{definition}\label{def: PSG and maximal loss}
Let $z_t\in\mathcal{Z}$ denote the state of the environment
at stage $t$ and $x\in\mathcal{X}$.  Define $\mu_{t,x} :=
\mathbf{E}\left[y_{t} \mid x, z_t\right]$, the expected
value of the feedback vector corresponding to arm
$x$ at state $z_t$. Then, the Pareto
suboptimality gap of arm $x$ is
\begin{equation}
    \Delta_t(x) := \inf \left\{ \epsilon \in \mathbb{R}_+ \mid \mu_{t,x} + \epsilon \not \prec \mu_{t,x'}, \forall x'\in\mathcal{X}\right\},
\end{equation}
and the maximal loss due to arm $x$ is
\begin{equation}
    \epsilon_t(x) := \inf \left\{\epsilon \in \mathbb{R}_+ \mid \mu_{t,x} +
        \epsilon \succ \mu_{t,x'}, \forall
    x' \in \mathcal{X}\right\}.\label{eq:epsilon}
\end{equation}
\end{definition}
We illustrate the difference between the two measures in Definition \ref{def: PSG and maximal loss}
using a simple example. Let $\mathcal{X} = \{A,B\}$ and at some stage $t$, $\mu_{t,A} = [ 0 \quad 1]^\top$ and $\mu_{t,B} = [ 2 \quad 0 ]^\top$.
According to Definition \ref{def: PSG and maximal loss},  $\Delta_t(A) = 0$ and $\Delta_t(B) = 0$, whereas $\epsilon_t(A) = 2$ and $\epsilon_t(B) = 1$.
In this example, neither arm's expected value of the feedback vector Pareto dominates the other, which the value of Pareto suboptimality gaps confirms.
However, pulling arm $A$ incurs
a loss of $2$ in the first objective, $\epsilon_t(A) = 2$, whereas pulling arm B
incurs a loss of
$1$ in the second objective, $\epsilon_t(B) = 1$. As a result, the algorithm must pick arm $B$ over $A$, or in general, it must pick the arm with the least value of maximal loss in individual objectives.

Next, we define Pareto regret and cumulative maximal
loss based on a sequence of Pareto suboptimality gaps and maximal losses, respectively.
\begin{definition} \label{def: PR and CML}
Let $h_T :=(z_1,x_1,y_1,\dots,z_{T},x_{T},y_T)$ be a history of states, actions, and feedback vectors over the learning horizon, $T$. Then, the Pareto regret (PR) and cumulative maximal loss (CML) corresponding to history $h_T$ are defined as
\end{definition}
\begin{equation}
    \mathrm{PR}(h_T) := \sum\limits_{t\in[T]}
    \Delta_t(x_t),\qquad\mbox{ and }\qquad
    \mathrm{CML}(h_T) := \sum\limits_{t\in[T]} \epsilon_t(x_t).
\end{equation}
The expert, who always picks a non-dominated arm, achieves a zero Pareto regret. In bandit algorithms, however, it is assumed that the agent does not know the distribution of the feedback vector, which is consistent with the assumption of unknown environment dynamics in blending controllers. Under this assumption, a \textit{sublinear} Pareto regret is often the best outcome that one can expect from a bandit algorithm \cite{lattimore2018bandit}, \textit{i.e.,}
\begin{equation}
    \lim\limits_{T\to\infty}\frac{1}{T}\mathrm{PR}(h_T) = 0.
    \end{equation}
\subsection{Problem statement} \label{subsec: problem statement}

We consider a learning agent who is provided with a set of pre-defined controllers and seeks to pick the controller whose $m$-dimensional feedback from the environment is not Pareto dominated by another controller's feedback. We assume that the agent has access to a feature mapping, $\psi$, which generates the context vector at each state of the environment. With the given controllers as the arms of a contextual multi-armed multi-objective bandit, we formulate two problems.
\begin{prob} \label{problem: 1}
    Design an online algorithm for blending controllers that achieves a sublinear Pareto regret.
\end{prob}
We address Problem \ref{problem: 1} under the following assumption on the structure of the feedback.
\begin{assumption}[\textsc{linear feedback with subgaussian noise}] \label{Assumption: linear}
At all stages $t\in[T]$, the context vector, $\Psi_t$, and feedback vector, $y_t$, satisfy
\begin{equation}
    \left(y_{t}\right)_i = \theta_{*,i} \cdot \Psi_t + \eta_t, \quad \forall i\in[m],
\end{equation}
where $\theta_{*, i} \in \mathbb{R}^d$ is an unknown coefficient vector and for a fixed $\sigma \in\mathbb{R}_+$, $\eta_t$ is conditionally $\sigma$-subgaussian, \textit{i.e.},
\begin{equation}
    \mathbf{E}\left[\exp{(\alpha \eta_t)} \mid
    \Psi_1,\dots,\Psi_t, \eta_{1}, \dots, \eta_{t-1}\right]
    \leq \exp{\left(\sigma^2\alpha^2/2\right)}, \quad
    \forall \alpha \in \mathbb{R}.
\end{equation}
\end{assumption}
\begin{prob} \label{problem: 2}
    Characterize an upper bound on the cumulative maximal loss of the proposed bandit algorithm for blending controllers.
\end{prob}
\section{Theoretical contributions}\label{sec:theory}

We propose Algorithm~\ref{alg:1} for blending
controllers. 
The algorithm is a bandit algorithm that picks an arm based on an upper confidence bound (UCB) index that it computes for each arm. 
Based on the UCB indices, the algorithm estimates the maximal loss in the individual
objectives corresponding to each arm.
It then picks
the arm with the least estimated loss. 
We compute the UCB indices using regularized least squares to exploit the results in \cite{abbasi2011improved}, in which it is shown that the accuracy of such estimates increases as time progresses.
In this section, we first show that the algorithm achieves a sublinear Pareto regret as the proposed solution to Problem \ref{problem: 1}. We then state our solution to Problem \ref{problem: 2}, wherein we establish an upper bound on the cumulative maximal loss of the algorithm.

For all objectives $i\in[m]$, and let $\hat{\theta}_{t,i}$ be
the $\ell^2$-regularized least-squares estimate of the
unknown coefficient vector $\theta_{*,i}$ with a
user-defined regularizing parameter $\lambda > 0$, \textit{i.e.,}
\begin{equation}\label{eq:thetahat}
    \hat{\theta}_{t,i} := \left(\Psi_{1:t}^\top \Psi_{1:t}^{}+ \lambda I\right)^{-1} \Psi_{1:t}^\top Y_{1:t}^{i},
\end{equation}
where $\Psi_{1:t}$ is the matrix with rows equal to $\Psi_1^\top,\Psi_2^\top,\dots,\Psi_t^\top$ and $Y_{1:t}^{i} = [(y_1)_{i},\dots,(y_t)_{i}]^\top$.
In Algorithm \ref{alg:1}, we implement a memory-efficient incremental implementation of \eqref{eq:thetahat}, see Appendix \ref{appendix: incremental implementation}.
The following lemma shows that, for all objectives $i\in[m]$,
the estimated coefficient vector, $\hat\theta_{t,i}$, is
statistically close to its true value, $\theta_{*,i}$.
\begin{lemma} [Theorem 2 in \cite{abbasi2011improved}] \label{lemma:ConfidenceSet}
Let Assumption \ref{Assumption: linear} hold and $\{\Psi_t\}_{t=1}^{\infty}$ be an $\mathbb{R}^d$-valued stochastic process, where $\Psi_t := \psi(z_t,x_t)$. For any $t\geq 1$, define
\begin{equation*}
    V_t := V_0 + \sum\limits_{s=1}^t \Psi_{s}{\Psi_{s}}^\top,
\end{equation*}
where $V_0=\lambda I_{d\times d}$, and $\lambda > 0$.
Furthermore, let $S,L\in\mathbb{R}_+$ and assume that, for all objectives $i\in[m]$ and stages $t\in\mathbb{N}$,
$\|\theta_{*,i}\|_2 \leq S$ 
and $\|\Psi_t\|_2 \leq L$. Then, for any
$\delta>0$, with probability at least $1-\delta$, the true coefficient vector, $\theta_{*,i}$, lies within ellipsoid $C_t^i$
defined as
\begin{equation}\label{eq:ct}
    C_{t}^i := \left\{ \theta \in \mathbb{R}^d :
    \left\|\theta - \hat{\theta}_{t,i} \right\|_{V_{t}} \leq \beta_t :=
\sigma
\sqrt{d\log\left(\frac{1+tL^2/\lambda}{\delta}\right)} +
\lambda^{1/2} S\right\}.
\end{equation}
\end{lemma}
At each stage $t\in[T]$, we use confidence ellipsoids $C_{t}^i$ to compute the UCB
indices corresponding to every arm $x\in\mathcal{X}$ and
objective $i\in[m]$. In particular, the UCB index of arm $x$ for objective $i$ is
\begin{equation}\label{eq:UCB}
    \hat{\mu}^i_{t,x} := \max\limits_{\theta \in C_{t}^i} \theta \cdot \psi(z_t, x).
\end{equation}
Equation \eqref{eq:UCB} is the evaluation of the support
function of ellipsoid $C_{t}^i$ at vector $\psi(z_t,
x)$. Using the closed-form solution for the support function
of an ellipsoid~\cite{boyd2004convex}, we can write

\begin{equation}\label{eq:UCB closed-form}
    \hat{\mu}^i_{t,x} = \left(\hat\theta_{t,i} +
        \frac{\beta_t V_t^{-1} \psi(z_t,x)}{\left\|\psi(z_t,x)\right\|_{V_t^{-1}}}
   \right) \cdot \psi(z_t,x) = \hat\theta_{t,i}\cdot \psi(z_t,x) +
        \beta_t \left\|\psi(z_t,x)\right\|_{V_t^{-1}}.
\end{equation}
\begin{algorithm}[t] 
\DontPrintSemicolon
  \KwInput{A set of controllers $\mathcal{X}$, feature mapping $\psi: \mathcal{Z}\times \mathcal{X} \mapsto \mathbb{R}^d$, regularization parameter $\lambda \ge \max\{1,L^2\}$, time horizon $T$}
  \textbf{Initialize }$V_0 = \lambda I_{d\times d}, \forall i\in[m]: \hat{\theta}_{1,i} = 0_{d\times1}, W_{0,i} = 0_{d\times 1}, \mathcal{O}_0 = \mathcal{X}$.\;
  \textbf{for} $t=1$ \textbf{to} $T$ \textbf{do}:\;
  \hspace{0.2in} Pull an arm $x_t \in \mathcal{O}_t$ uniformly at random.\;
  \hspace{0.2in} Observe the state of the environment $z_t$ and feedback $y_t$, and let $\Psi_t := \psi(z_t,x_t)$.\;
  \hspace{0.2in} Update $V_{t} = V_{t-1} + \Psi_t {\Psi_t}^\top$.\;
  \hspace{0.2in} \textbf{for} $i=1$ \textbf{to} $m$ \textbf{do}:\;
  \hspace{0.4in} Update $W_{t,i} = W_{t-1,i} + (y_t)_i \Psi_t$.\;
  \hspace{0.4in} Compute $\hat{\theta}_{t,i} = V_{t}^{-1} W_{t,i}$.\;
  \hspace{0.4in} Compute $C_{t}^i$ by \eqref{eq:ct}.\;
  \hspace{0.2in} For each arm $x\in\mathcal{X}$, compute the UCB index, $\hat\mu_{t,x}$, by \eqref{eq:UCB closed-form}.\; 
      \hspace{0.2in} For each arm $x\in\mathcal{X}$, compute the estimated maximal loss, $\hat\epsilon_t$, by \eqref{eq:epsilon hat definition}.\;\label{step:arm}
  \hspace{0.2in} Update $\mathcal{O}_{t+1} = \mathop{\mathrm{argmin}}_{x\in\mathcal{X}} \hat \epsilon_t(x)$.
\caption{Blending controllers algorithm}\label{alg:1}
\end{algorithm}
Subsequently, Algorithm~\ref{alg:1} estimates the maximal
losses using the computed UCB indices as
\begin{equation} \label{eq:epsilon hat definition}
    \hat \epsilon_t(x) := \inf\left\{ \epsilon \in \mathbb{R}_+ \mid \hat \mu_{t,x} + \epsilon \succ \hat \mu_{t,x'}, \forall x'\in\mathcal{X}\right\}, \quad \forall x\in\mathcal{X},
\end{equation}
and picks the arm with the least estimated maximal loss. In Lemma \ref{lemma: epsilon and dominance}, we show that such a decision rule between the arms picks an arm from the set of non-dominated arms.

\begin{lemma}
    \label{lemma: epsilon and dominance}
    In Algorithm \ref{alg:1}, let $x_t\in\mathcal{X}$ denote the arm picked by the algorithm at stage $t\in[T]$. Then, the UCB index corresponding to arm $x_t$ is not Pareto dominated by any other arm's UCB index.
\begin{proof}
See Appendix \ref{appendix: proof of lemma 2}
\end{proof}
\end{lemma}
We now state the main theorem of this paper, which alongside the algorithm itself, solves Problem \ref{problem: 1}.
\begin{theorem}[\textsc{Sublinear Pareto regret}] \label{theorem: Sublinear Pareto regret}
Let Assumption \ref{Assumption: linear} hold and $S,L\in\mathbb{R}_+$. Assume that, for all objectives $i\in[m]$ and all stages $t\in[T]$, $\|\theta_{*,i}\|_2 \leq S$ and $\|\Psi_t\|_2 \leq L$.
Then, with high probability $1-\delta$, Algorithm \ref{alg:1} satisfies
\begin{equation*}
    \mathrm{PR}(h_T) \leq \mathcal{O}\left(\sqrt{T}\log(T)\right),
\end{equation*}
where $h_T$ is the history of states, actions and feedback vectors over the learning horizon $T$.
\begin{proof}
See Appendix \ref{appendix: sublinear pareto regret}.
\end{proof}
\end{theorem}


We now establish an upper bound on the average cumulative maximal loss
of the algorithm during
execution.
We use the upper bound as an additional performance measure to Pareto regret in order to bridge the gap between the multi-objective performance and performance in individual objectives. In the following theorem, we state our solution to Problem \ref{problem: 2}.
\begin{theorem}\label{theorem:PRD}
Assume the same settings as in Theorem \ref{theorem:
Sublinear Pareto regret}. Then, for any $T\in\mathbb{N}$, the average cumulative maximal
loss of Algorithm \ref{alg:1} satisfies
\begin{equation} \label{eq:average loss UB}
    0 \le \frac{1}{T} \mathrm{CML}(h_T)   \le \frac{1}{T} \sum\limits_{t=1}^{T} \hat \epsilon_t(x_t) + \mathcal{O}(1/\sqrt{T}),
\end{equation}
where $h_T$ is the history of states, actions and feedback vectors until stage $T$, and $x_t$ is the chosen arm at stage $t\in[T]$.
\end{theorem} 
\begin{proof}
See Appendix \ref{appendix: theorem 2}.
\end{proof}
Observe that the upper bound in \eqref{eq:average loss UB}
depends on the computed estimated maximal loss values and a diminishing term of the order $\mathcal{O}(1/\sqrt{T})$. Thus, the proposed
upper bound does not impose any additional computational
complexity in computation.

\section{Experiments}
\label{sec:numerical-examples}
\begin{figure}[t]
    \centering
    \begin{subfigure}[t]{0.26\linewidth}
        \centering
        \includegraphics[width = \linewidth]{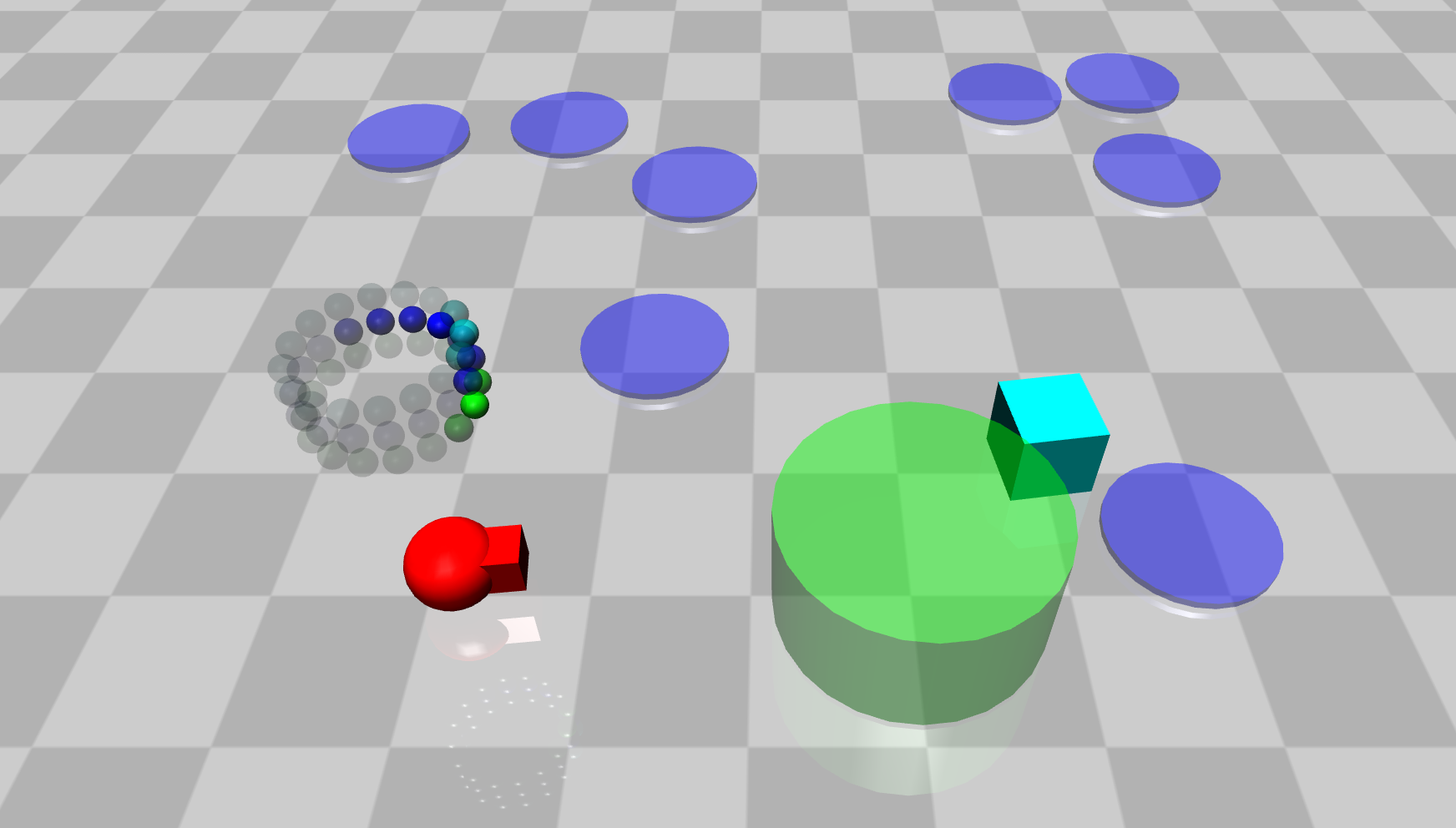}
        \caption{\emph{Point-Goal} }\label{fig:pointgoal}
    \end{subfigure}
    \quad
    \begin{subfigure}[t]{0.26\linewidth}
        \centering
        \includegraphics[width=\linewidth]{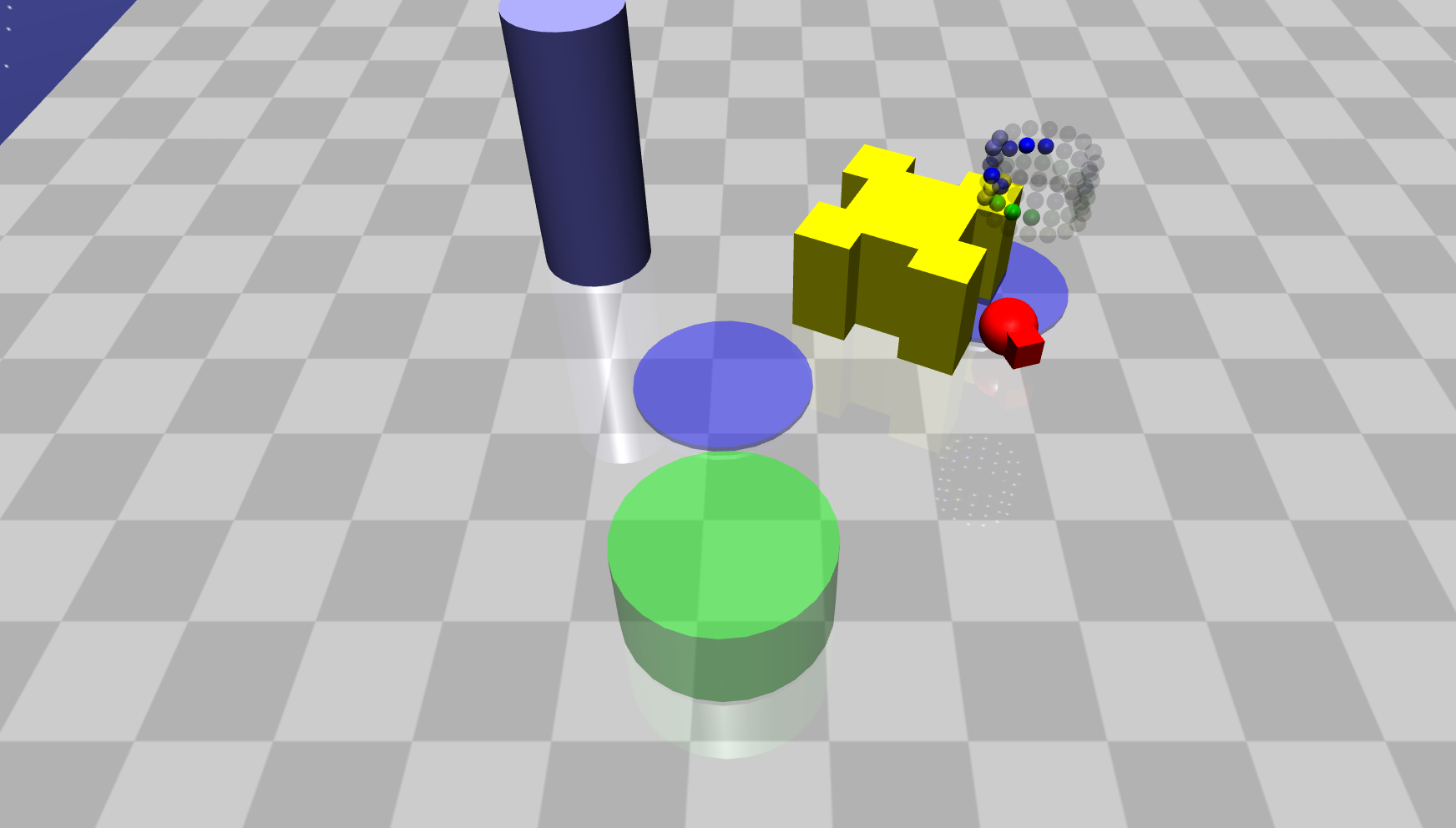}
        \caption{\emph{Point-Push}}\label{fig:pointpush}
    \end{subfigure}
    \quad 
    \begin{subfigure}[t]{0.26\linewidth}
        \centering
        \includegraphics[width=\linewidth]{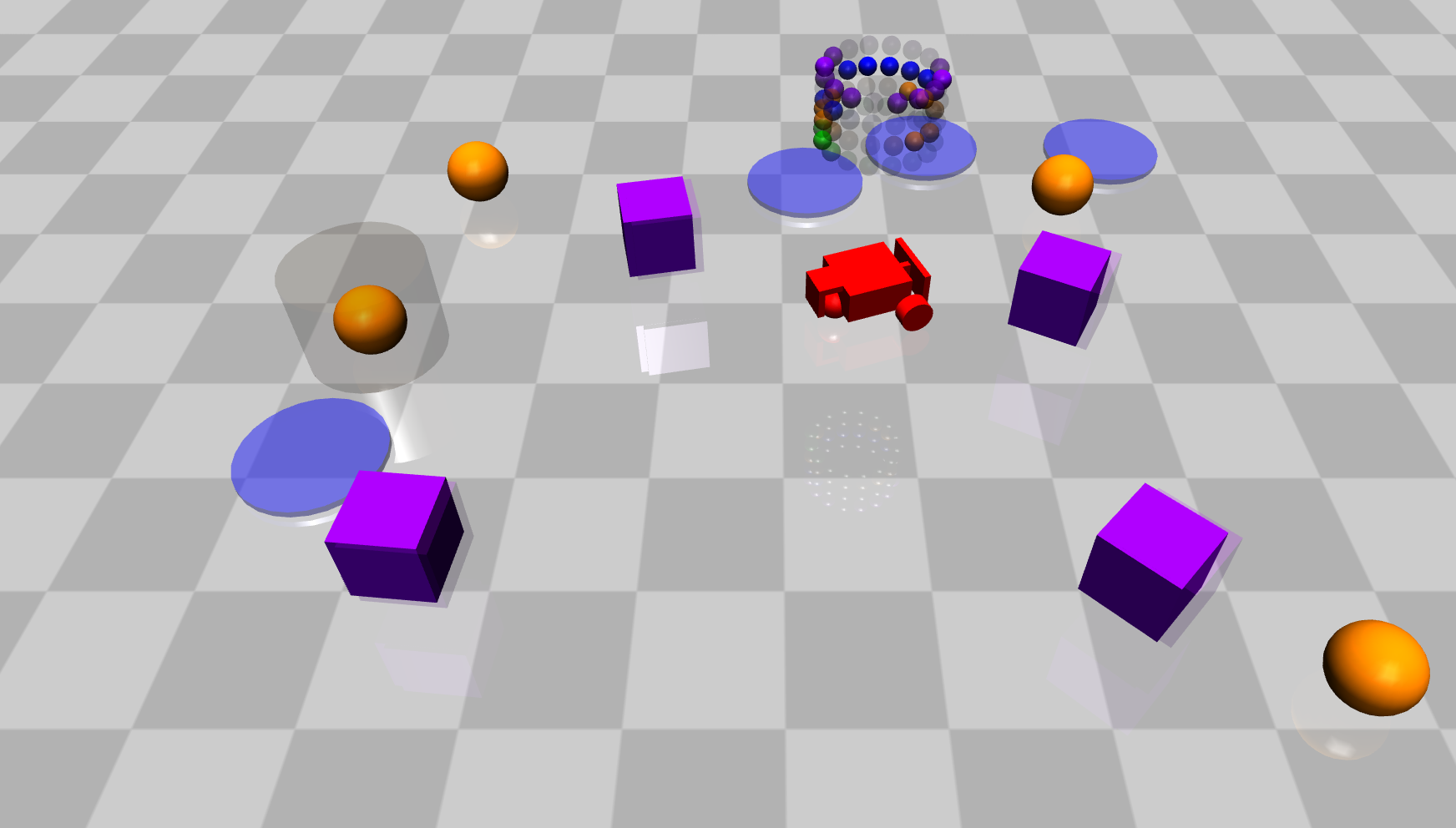}
        \caption{\emph{Car-Button}}\label{fig:carbutton}
    \end{subfigure}
    \caption{The selected Safety Gym environments for the experiments in Section \ref{sec:numerical-examples}.}
    \label{fig:safety-gym-envs}
\end{figure}
In this section, we use the Safety Gym suite \cite{raybenchmarking} to demonstrate the proposed algorithm's effectiveness in blending controllers in practice.
We find Safety Gym a suitable testbed because (i) state-of-the-art reinforcement learning algorithms that respect safety have been benchmarked in all of its varied environments, and (ii) similar to this paper, it uses the framework of auxiliary cost functions to enforce safety requirements. In our experiments, we chose the Point-Goal, Point-Push, and Car-Button Safety Gym environments, with each environment described as follows:
\begin{itemize}
    \item \emph{Point-Goal.} A robot with two actuators, one that sets the thrust and the other sets the angle, has to reach the green zone in Figure~\ref{fig:pointgoal}, while staying clear from the dangerous areas highlighted as blue circles. The robot itself is depicted in red.
    \item \emph{Point-Push.} The same robot has to navigate the yellow box in Figure~\ref{fig:pointpush} to the green zone. In addition to the safety specifications in the Point-Goal environment, the agent has to avoid the erected pillars in the environment.
    \item \emph{Car-Button.} A car with two independently actuated front wheels and a free-rolling rear wheel has to press the orange button that is highlighted as in Figure~\ref{fig:carbutton}. The agent has to avoid the purple moving boxes as a more sophisticated safety requirement.
    
\end{itemize}

We now describe and justify our choices of $\mathcal{X}$, the set of input controllers, $\psi$, the feature mapping, and the hyper parameters, $\lambda$, $L$, $S$ and $\sigma$ that are required by Algorithm \ref{alg:1}.

\textbf{Input controllers $\mathbf{\mathcal{X}}$.} For each environment, we use the Safety Gym benchmark suite~\cite{raybenchmarking} to identify the suitable choice of the safe and the performant controllers.
We use reinforcement learning algorithms such as trust region policy optimization (TRPO)~\cite{schulman2015trust} and proximal policy optimization (PPO)~\cite{schulman2017proximal} to generate the performant controller. According to the above benchmarking, these controllers show a superior performance in terms of their total reward; however, they perform poorly in terms of safety.
We generate the safe controller using constrained reinforcement learning algorithms such as constrained policy optimization (CPO)~\cite{achiam2017constrained} or a combination of Lagrangian methods with PPO and TRPO. In contrast to TRPO and PPO, these algorithms have been shown to meet their safety requirements in the sense that their average cost is below a fixed threshold; however, they are conservative in performance.


\captionsetup[figure]{font=normal,skip=0pt}
\begin{figure*}[t]
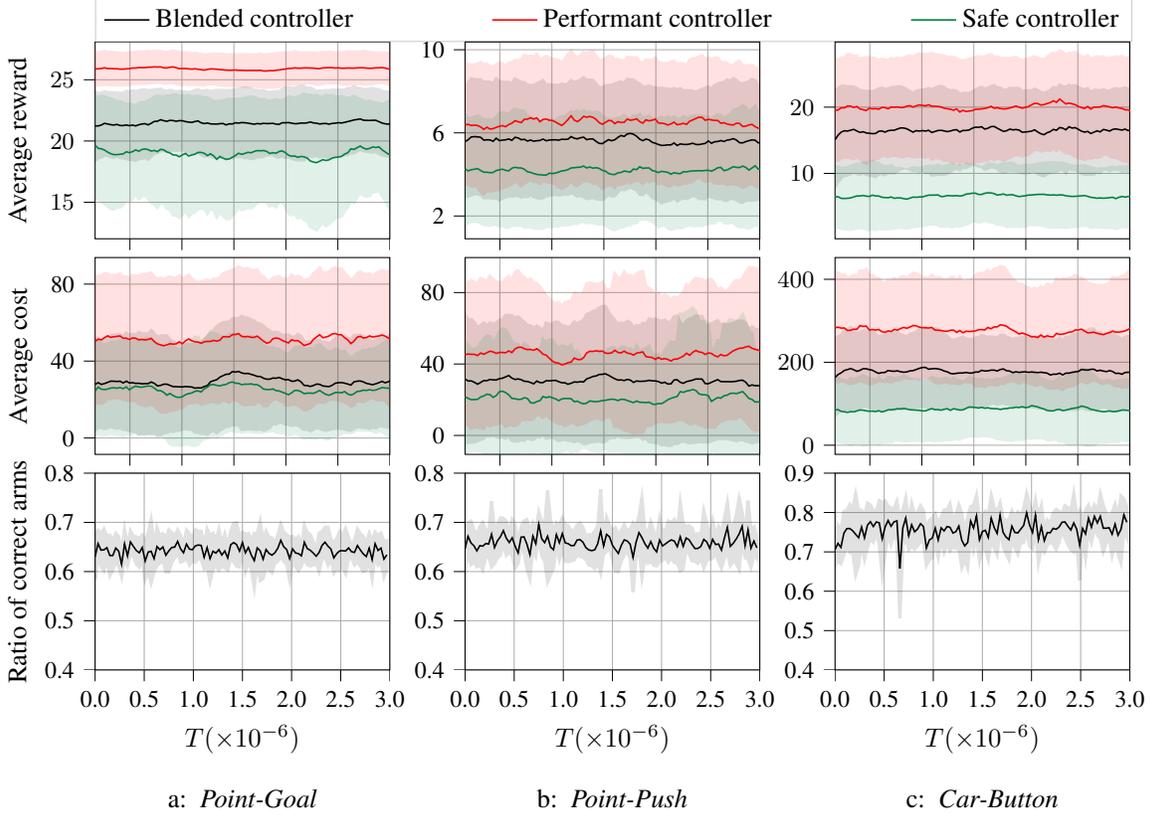

    \centering
    \hspace*{-1cm}
\pgfplotsset{every tick label/.append style={font=\small}}


    \renewcommand\thesubfigure{(\alph{subfigure})}
    \caption{Simulation results of testing Algorithm 1 in the selected Safety Gym environments in Section \ref{sec:numerical-examples}. The safe controllers used in a, b, and c are derived from PPO-Lagrangian, CPO, and TRPO-Lagrangian algorithms, respectively. The performant controllers are learned by PPO in a and TRPO in both b and c. The confidence clouds correspond to the standard deviation of each data point.}
    
    \label{fig:metrics-blended}
\end{figure*}

\textbf{Feature mapping $\mathbf{\psi}$.} The above choices of the safe and the performant controllers belong to the family of deep reinforcement learning algorithms~\cite{mnih2015human}, which use a neural network to estimate the action values at each state of the environment. We incorporate the underlying neural networks as the feature mapping to generate the context at each state. Precisely, each controller uses a separate feedforward multilayer perceptron network of size $(256,256)$ with $\tanh$ activation functions. In practice, the trained neural networks provide accurate approximations of the next-step reward and cost, and therefore, we consider that Assumption \ref{Assumption: linear} holds.

\textbf{Hyperparameters.} Based on the observed values of the output layer of the neural networks, we set $L$, the upper bound on the 2-norm of the feature vector, to 1, and we subsequently set $\lambda = 1$. Since the neural networks directly estimate the reward and cost, we expect the true coefficient vector to be at the vicinity of either $[1 \quad 0]^\top$ or $[0 \quad 1]^\top$; therefore, we set $S$, the upper bound on the 2-norm of the true coefficient vector $\theta_{*}$, to $1.5$. We choose a small value for $\sigma$ because we expect the variance of the mismatch noise to be small. In particular, we set $\sigma = 0.1$.


We visualize the results of each experiment with plots as depicted in Figure \ref{fig:metrics-blended}. The first row compares the average reward corresponding to the performant, the safe, and the blended controller. In the second row, we compare the blended controller's safety with its safe and performant counterparts by comparing their average cost. Finally, in the third row we evaluate the rate at which the algorithm employs the correct controller, \textit{i.e.,} it successfully avoids the controller with Pareto dominated reward and cost feedback. In order to find such a metric, at each decision step, we fix the environment behavior and separately employ each controller to reveal their true reward and cost. Then, we are able to establish the Pareto dominance relationship amongst them. Each data point in Figure \ref{fig:metrics-blended} represents an average of the metrics of 30 episodes, with each episode length fixed at 1,000 iterations.


As desired, the numerical results suggest that the proposed algorithm for blending controllers is an effective method of finding a controller that improves the safety of a given performant controller and the performance of a given safe controller. An analysis of the standard deviation across batches of 30,000 iterations supports the claim that the average cost and reward of the blended controller lie in-between the given performant and safe controllers. Additionally, the ratio at which the correct controller is picked is always above $0.5$, which indicates that compared to naively switching between the safe and the performant controller at random, Algorithm \ref{alg:1} performs significantly better.

We used a desktop computer with Intel Core $i9$-$9900$, \SI{3.10}{GHz} $\times 16$ processors and \SI{32}{Gb} of RAM for the experiments. Each iteration of the \emph{Point-Goal}, \emph{Point-Push}, and \emph{Car-Button} environments took \SI{2}{ms}, \SI{3}{ms}, and \SI{5}{ms}, respectively. See \cite{github} for all codes and datasets used in this section.

\section{Related Work}
\label{sec:related_work}

We review existing work on constrained reinforcement
learning algorithms that can enforce safety. 
We share the use of auxiliary costs to model safety with
constrained reinforcement learning algorithms, such as
CPO and PPO-Lagrangian. Although these
algorithms show promise in the
\textit{Safety Gym} benchmark suite, even
slight changes to safety specifications require learning a
new policy from scratch. In contrast, the proposed modular approach of blending controllers allows the decision-maker to simply blend the outdated policy with a new policy that satisfies the new constraint.

In \cite{verma2019imitation}, the authors propose a \emph{programmatic} reinforcement learning algorithm that shares the same modular methodology as employed in this paper.
Although the proposed algorithm improves upon the performance of the input controllers, it does not consider safety as an objective, which is the main concern of this paper.
In another approach, the \textit{simplex architecture}~\cite{sha2001using,bak2014real, lee2005dependable, bak2011sandboxing}, a decision logic for controlling a plant is developed using a safe and a performant controller.
In these works, the decision logic is computed by either a Lyapunov function or a reachability analysis.
These approaches require prior knowledge of the environment dynamics, whereas we do not assume any such prior knowledge.


The work in \cite{cheng2019control} introduces a regularization scheme for a policy gradient reinforcement learning algorithm, such that the policy updates occur at the vicinity of a given reference controller. The introduced methodology splits the environment dynamics into a sum of a known and an unknown environment dynamics. Then, the authors show that the learned policy inherits the Lyapunov stability guarantees that the reference controller provides for the known environment dynamics. However, these guarantees are not applicable for an environment with fully unknown dynamics. In contrast, the safety modeling employed in this paper is specialized for such environments.

Finally, we review \textit{shared control protocols} ~\cite{broad2018learning, jansen2017synthesis,dragan2013policy, benloucif2019cooperative, ezeh2017comparing, jain2018recursive}, wherein a robot's commands are blended with its human user. Shared control protocols and blending controllers have the same goal of balancing the safety and performance of their given controllers. The algorithms in \cite{benloucif2019cooperative, ezeh2017comparing, jain2018recursive} use a weighted sum of the given controllers' outputs, whereas in this work, we switch between the given controllers. As a result, we are not restricted to continuous control signals. 
In \cite{broad2018learning}, the environment dynamics are unknown but assumed to be linear.
Finally, the authors in \cite{jansen2017synthesis} use semi-definite programming for blending controllers, which requires prior knowledge of the environment dynamics and the confidence level of each controller. 
In contrast to \cite{broad2018learning} and \cite{jansen2017synthesis}, the approach proposed in this paper is model-free. 


\section{Conclusions and future work} \label{sec:conclusions}
We developed a contextual multi-armed multi-objective bandit algorithm to solve the problem of blending controllers.
The algorithm achieves a sublinear Pareto regret, which characterizes its performance measure.
We also derived an upper bound on the algorithm's cumulative maximal loss, which shows how much cost or reward the algorithm sacrifices while execution.
We empirically demonstrated the algorithm's performance in the Safety Gym suite. The simulation results show that the algorithm succeeds in learning the appropriate switching strategy in the sense that the algorithm's total reward is higher than its given safe controller, and it accumulates less cost than the given performant controller.\par

For the future work, we are interested in considering adversarial environments to investigate when the Pareto regret achieved by Algorithm \ref{alg:1} is optimal. We will also consider a generalized UCB-based decision rule beyond the comparison of the estimated maximal losses. There, we are interested in the sufficient conditions under which the algorithm maintains the sublinearity of its Pareto regret.

\section{Broader impacts} \label{sec:broader impacts}
In this section, we retreat from the technicalities into the broader societal impacts of this work. We formulated the problem of blending controllers to address the safety concerns that arise with AI systems, for example safety concerns in autonomous driving \cite{koopman2017autonomous}. Blending controllers improves the safety of its input controllers while it is blind to the algorithm that drives them. Such blindness helps resolving the privacy concerns that may keep industrial companies from sharing their breakthroughs in safe control algorithms. Blending state-of-the-art safe controllers while protecting the privacy of their underlying algorithms may prompt industrial companies to at least share the output of their controllers in the name of protecting human lives.

On the other hand, even the most innocent intentions may lead to negative consequences. Multiple studies have found human users trusting automated systems in inappropriate circumstances, see \cite{lewis2018role} and references therein. Improving the safety of artificial intelligent systems may exacerbate such over-reliance. We emphasize that although blending controllers improves the safety of the overall system, its level of safety depends on its input safe controllers. Therefore, it is crucial that potential users be warned to perceive any improvements in automation safety as 
a lower chance of safety hazards and not an elimination of its vulnerabilities.



\bibliographystyle{plain}

\appendix
\section*{Appendices}
\section{Incremental $\mathbf{\ell^2}$-regularized least squares}\label{appendix: incremental implementation}
Recall that at stage $t$, the $\ell^2$-regularized least-squares estimate of $\theta_{*,i}$, the coefficient vector corresponding to objective $i$, is
\begin{equation}\label{eq:appendix a theta hat}
    \hat{\theta}_{t,i} := \left(\Psi_{1:t}^\top \Psi_{1:t}^{}+ \lambda I\right)^{-1} \Psi_{1:t}^\top Y_{1:t}^{i}.
\end{equation}
Let $V_0 := \lambda I$, $\forall i\in[m]: W_{0,i} := 0_{d\times 1}$, and for all $t \ge 1$,
\begin{equation}
    V_t := \Psi_{1:t}^\top \Psi_{1:t}^{}+ \lambda I \quad \text{and} \quad W_{t,i} := \Psi_{1:t}^\top Y_{1:t}^{i}.
\end{equation}
Then, at stage $t+1$, we can write
\begin{equation} \label{eq: appendix a eq 1}
    V_{t+1} = \left[\begin{array}{cccc} \Psi_1 & \dots & \Psi_t & \Psi_{t+1}  \end{array} \right] \left[ \arraycolsep=1.4pt\def\arraystretch{1.8}\begin{array}{c} \Psi_1^\top \\ \vdots \\ \Psi_t^\top \\  \Psi_{t+1}^\top  \end{array}\right] + \lambda I = V_t + \Psi_{t+1}^{} \Psi_{t+1}^\top,
\end{equation}
and
\begin{equation}\label{eq: appendix a eq 2}
    W_{t+1,i} =  \left[\begin{array}{cccc} \Psi_1 & \dots & \Psi_t & \Psi_{t+1}  \end{array} \right] \left[ \arraycolsep=1.4pt\def\arraystretch{1.8}\begin{array}{c} (y_1)_i \\ \vdots \\ (y_t)_i \\ (y_{t+1})_i  \end{array}\right] = W_{t,i} + (y_{t+1})_i \Psi_{t+1}.
\end{equation}
By \eqref{eq: appendix a eq 1} and \eqref{eq: appendix a eq 2}, we arrive at
\begin{equation}
    \hat{\theta}_{t,i} = V_t^{-1} W_{t,i}, \quad \forall i\in[m],
\end{equation}
which no longer requires storing all the values of $\Psi$ and $y$.

\section{Proof of Lemma \ref{lemma: epsilon and dominance}} \label{appendix: proof of lemma 2}
Assume that there exists an arm, $x'\in\mathcal{X} \setminus
\mathrm{argmin}_{x\in\mathcal{X}}  \hat \epsilon(x)$, whose
UCB index, $\hat \mu_{t,x'}$, Pareto dominates that of the
picked arm, $x_t$. Then, by the definition of $\hat\epsilon$
in \eqref{eq:epsilon hat definition}, it follows that
$\hat\epsilon(x') = 0$, and therefore $x' \in
\mathrm{argmin}_{x\in\mathcal{X}}  \hat \epsilon(x)$, which
is in contradiction with the initial assumption.

\section{Proof of Theorem \ref{theorem: Sublinear Pareto regret}} \label{appendix: sublinear pareto regret}

Let $x_t\in\mathcal{X}$ denote the arm that is chosen by the algorithm and $z_t\in\mathcal{Z}$ denote the state of the environment at stage $t$. For all objectives $i\in[m]$, let
\begin{equation} \label{eq:ThetaTilde}
     \hat{\mu}_{t,x_t}^i = \Tilde{\theta}_{t,i} \cdot \psi(z_t,x_t), \quad \text{with} \quad \Tilde{\theta}_{t,i} := \mathop{\text{argmax}}\limits_{\theta \in \mathcal{C}_t^i} \theta \cdot \psi(z_t,x_t).
\end{equation}
The algorithm chooses amongst the set of arms whose UCB indices are not Pareto dominated by that of any other arm; therefore, for each arm $x\in\mathcal{X}$, there exists an objective $j\in[m]$ such that $\hat{\mu}_{t,x_t}^j \geq \hat{\mu}_{t,x}^j$.
By Lemma \ref{lemma:ConfidenceSet}, we have that with probability at least $1-\delta$,
\begin{equation*}
    \hat{\mu}_{t,x}^j = \max\limits_{\theta \in \mathcal{C}_{t}^j} \theta \cdot \psi(z_t,x) \geq \theta_{*,j} \cdot \psi(z_t,x).
\end{equation*}
Hence, with probability at least $1-\delta$, for all arms $x\in\mathcal{X}$,
\begin{equation} \label{eq: IndexJ}
    {\Tilde{\theta}}_{t,j} \cdot \psi(z_t,x_t) \geq \theta_{*,j} \cdot \psi_t(z_t,x).
\end{equation}
We now consider the case in which the algorithm has picked the wrong arm, \textit{i.e.,} there exists an arm $x'\in\mathcal{X}$ such that $\theta_{*,j} \cdot \psi(z_t,x_t) < \theta_{*,j} \cdot \psi(z_t,x')$.
Then, for any $t\ge1$, we can write
\begin{align}
    \theta_{*,j} \cdot \psi(z_t,x') -  \theta_{*,j} \cdot \psi(z_t,x_t) &\leq \tilde \theta_{t,j} \cdot \psi(z_t,x_t) - \theta_{*,j} \cdot \psi(z_t,x_t) \nonumber \\
    & = (\tilde \theta_{t,j} - \hat \theta_{t,j})\cdot \psi(z_t,x_t) + (\hat \theta_{t,j} - \theta_{*,j})\cdot \psi(z_t,x_t) \nonumber\\
    & \leq \left(\|\tilde \theta_{t,j} - \hat \theta_{t,j}\|_{V_{t}} + \|\hat \theta_{t,j} - \theta_{*,j}\|_{V_{t}}\right)\|\psi(z_t,x_t)\|_{V^{-1}_t} \label{eq:Holder}\\
    & \leq 2\beta_{t} \|\psi(z_t,x_t)\|_{V^{-1}_{t}} \le 2\beta_{T} \|\psi(z_t,x_t)\|_{V^{-1}_{t}}. \nonumber
\end{align}
The first inequality is a result of \eqref{eq: IndexJ}. Inequality \eqref{eq:Holder} holds because of the H{\"o}lder's inequality. Finally the last inequality follows from the fact that $\beta_t$ is an increasing function of $t$. Therefore, 
\begin{equation}\label{eq:optimalitygap}
  \theta_{*,j}\cdot \psi(z_t,x') - \theta_{*,j}\cdot \psi(z_t,x_t) \leq 2\beta_{T} \|\Psi_t\|_{V^{-1}_{t}}.
\end{equation}
Notice that the upper bound in \eqref{eq:optimalitygap} is independent of index $j$ and arm $x'$. Therefore, the Pareto suboptimality gap of $x_t$, $\Delta_t x_t$, is also upper bounded by $2\beta_{T} \|\Psi_t\|_{V^{-1}_{t}}$.
By the Cauchy-Schwarz inequality, we have that for all $i\in[m]$ and all $x\in\mathcal{X}$, $\theta_{*,i} \cdot \psi(z_t,x) \leq \|\theta_{*,i}\|_2 \|\Psi(z_t,x)\|_2 \leq SL$.
Let $(a \wedge b):=\max(a,b)$. Then, we can write
\begin{equation}\label{eq:bound}
      \Delta_t x_t  \leq \left( 2SL \wedge 2\beta_{T}
    \|\Psi_t\|_{V^{-1}_t}\right) = 2\beta_{T} \left(\frac{SL}{\beta_{T}} \wedge  \|\Psi_t\|_{V^{-1}_t}\right) \leq 
    2\beta_{T} \left(1 \wedge  \|\Psi_t\|_{V^{-1}_t}\right).
\end{equation}
Taking the sum of both sides of \eqref{eq:bound} and using the Cauchy-Schwarz inequality, we can write
\begin{equation*}
    \mathrm{PR}(h_T) \leq  \sqrt{8\beta_{T}^2T\sum\limits_{t=1}^T\left(1 \wedge  \|\Psi_t\|_{V^{-1}_t}^2\right)}.
\end{equation*}
Finally, by Lemma 11 in \cite{abbasi2011improved} we have that with probability at least $1-\delta$,
\begin{equation*}
    \mathrm{PR}(h_T) \leq 8 \left(\sigma \sqrt{d\log\left(\frac{1+TL^2/\lambda}{\delta}\right)} + \lambda^{1/2} S\right)^2\sqrt{2Td\log(\lambda + TL/d)},
\end{equation*}
which concludes the proof.

\section{Proof of Theorem \ref{theorem:PRD}} \label{appendix: theorem 2}
We start off the proof with reformulating $\epsilon$ in
\eqref{eq:epsilon}. For any arm $x\in\mathcal{X}$, we can write
\begin{equation}
    \epsilon_t(x) = \max \left\{ \max\limits_{k\in[m]}\max\limits_{x'\in\mathcal{X}} \left(\mu_{t,x'}^k - \mu_{t,x}^k\right), 0 \right\}.
\end{equation}
Analogously, for all arms $x\in\mathcal{X}$, we have that
\begin{equation}
    \hat\epsilon_t(x) = \max \left\{ \max\limits_{k\in[m]}\max\limits_{x'\in\mathcal{X}} \left(\hat\mu_{t,x'}^k - \hat\mu_{t,x}^k \right), 0 \right\}.
\end{equation}
Let $x_t\in\mathcal{X}$ be the arm that the algorithm picks at stage $t$. Then, $\hat\epsilon(x_t) = \min_{x\in\mathcal{X}} \hat\epsilon(x)$. Then, for all arms $x\in\mathcal{X}$ and objectives $j\in[m]$, there exist an arm $x'\in\mathcal{X}$ and objective $k\in[m]$ such that
\begin{equation} \label{eq:appendix C eq 3}
    \hat\mu_{t,x}^j -\hat \mu_{t,x_t}^j \le \hat\mu_{t,x'}^k - \hat\mu_{t,x_t}^{k}.
\end{equation}
By Lemma \ref{lemma:ConfidenceSet}, with high probability $1-\delta$, it holds that $(\hat\mu_{t,x})_j \ge (\mu_{t,x})_j$. Rearranging \eqref{eq:appendix C eq 3}, we arrive at
\begin{equation} \label{eq: appendix C eq 4}
    \mu_{t,x}^j + \hat \mu_{t,x_t}^{k} \le \hat\mu_{t,x_t}^j + \hat\mu_{t,x'}^{k}.
\end{equation}
For all arms $x\in\mathcal{X}$ and objectives $j\in[m]$, we can write
\begin{align}\label{eq: appendix C eq 5}
    \mu_{t,x}^j - \mu_{t,x_t}^j &\le \hat\mu_{t,x_t}^j + \hat\mu_{t,x'}^{k} - \hat\mu_{t,x_t}^{k} - \mu_{t,x_t}^j \nonumber\\
    &= (\tilde\theta_{t,j} - \theta_{*,j}) \cdot \Psi_t + \hat\mu_{t,x'}^{k} - \hat\mu_{t,x_t}^{k} \nonumber\\
    &\le 2\beta_{T} \|\Psi_t\|_{V^{-1}_{t}} + \hat\mu_{t,x'}^{k} - \hat\mu_{t,x_t}^{k}\nonumber\\
    &\le 2\beta_{T} \|\Psi_t\|_{V^{-1}_{t}} + \max\limits_{k\in[m]}\max\limits_{x'\in\mathcal{X}} \left(\hat\mu_{t,x'}^k- \hat \mu_{t,x_t}^k \right)\nonumber \\
    &\le 2\beta_{T} \|\Psi_t\|_{V^{-1}_{t}} + \hat \epsilon(x_t),
\end{align}
where the first inequality is resulted from \eqref{eq: appendix C eq 4}, and the second inequality follows from the same argument as in the proof of Theorem \ref{theorem: Sublinear Pareto regret} in Appendix \ref{appendix: sublinear pareto regret}.
Equation \eqref{eq: appendix C eq 5} holds for all arms $x$ and objectives $j$, hence
\begin{equation}\label{eq:appendix C eq 6}
    0 \le \epsilon_t(x_t)   \le 2\beta_{T} \|\Psi_t\|_{V^{-1}_{t}} + \hat \epsilon_t(x_t).
\end{equation}
Taking the average of both sides of \eqref{eq:appendix C eq 6}, followed by the results in Theorem \ref{theorem: Sublinear Pareto regret}, we arrive at
\begin{equation}
    0 \le \frac{1}{T} \sum\limits_{t=1}^{T} \epsilon_t(x_t)
    \le \frac{1}{T}{\sum\limits_{t=1}^{T} \hat
    \epsilon_t(x_t)} + \mathcal{O}(1/\sqrt{T}).
\end{equation}
This completes the proof.

\end{document}